\documentclass{llncs}
\usepackage[T1]{fontenc}
\usepackage[utf8]{inputenc}
\usepackage{amsmath}
\usepackage{amsfonts}
\usepackage{amssymb}
\usepackage{graphicx}

\begin{document}
	
	\title{From de Bruijn graphs to variation graphs 
		-- relationships between pangenome models}
	
	\author{Adam Cicherski \orcidID{0009-0007-5707-8967} \and Norbert Dojer\orcidID{0000-0001-5653-1167}}
	
	\institute{University of Warsaw, Warsaw, Poland}
	
	\maketitle
	
	\begin{abstract}
		Pangenomes serve as a framework for 
		joint analysis of genomes of related organisms.
		Several pangenome models were proposed, 
		offering different functionalities, 
		applications provided by available tools, their efficiency etc.
		Among them, two graph-based models
		are particularly widely used:
		variation graphs and de Bruijn graphs.
		
		In the current paper we propose an axiomatization
		of the desirable properties
		of a graph representation of a collection of strings.

		We show the relationship between 
		variation graphs satisfying these criteria 
		and de Bruijn graphs.
		This relationship can be used to 
		efficiently build a variation graph 
		representing a given set of genomes,
		transfer annotations between both models,
		compare the results of analyzes based on
		each model etc.
		
		\keywords{pangenome \and de Bruijn graphs \and variation graphs.}
	\end{abstract}
	
	\section{Introduction}

	The term pangenome was initially proposed as a single data structure 
	for joint analysis of a group of bacterial genes \cite{tettelin_2005}. 
	In the presence of a variety of whole genome sequences available it has evolved 
	and currently it refers to a model of joint analysis of genomes of related organisms. 
	Numerous pangenome models were proposed, ranging from collections of unaligned sequences 
	to sophisticated models that require complex preprocessing of sequence data 
	\cite{computationalpangenomicsconsortium_2018,eizenga_2020,paten_2017}.
	The models significantly differ in their properties, 
	applications provided by available tools, their efficiency, 
	and even the level of precision of the definition of the optimal graph representing given genomes.

	De Bruijn graphs (dBG) have nodes uniquely labeled with
	$k$-mers and edges representing their overlaps of length $k-1$. 
	Since the structure of a dBG for a given set of genomes 
	is strictly determined by the parameter $k$,
	its construction is straightforward and may be performed in linear time. 
	Applications of dBG-based pangenome models include read mapping
	\cite{heydari_2018,limasset_2016}, variant calling \cite{iqbal_2012}, 
	taxonomic classification and quantification of metagenomic sequencing data \cite{schaeffer_2017} 
	and even pangenome-wide association studies \cite{manuweera_2019}.
	
	Variation graphs (VG) have nodes labeled with DNA sequences of arbitrary length.
	Concatenations of labels of consecutive nodes on paths
	form represented genomic sequences.
	Such structure allows to avoid the redundancy of the dBG representation,
	where a single residuum occurs in labels of several nodes. 
	Consequently, this model provides an intuitive common coordinate system,
	making it convenient to annotate, 
	which is crucial in many applications.
	The basis for a wide spectrum of VG-based analyzes was provided by indexing methods 
	for efficient subsequence representation and searching
	\cite{durbin_2014,sirn_2014,sirn_2017}.
	Then sequencing read mapping algorithms were proposed
	\cite{garrison_2018,rautiainen_2020}, 
	opening the door for variant calling and genotyping tools
	\cite{eggertsson_2017,eggertsson_2019,hickey_2020}. 
	Moreover, variation graphs were applied to the inference of precise haplotypes 
	\cite{baaijens_2019,baaijens_2019a}
	and genome graph-based peak calling \cite{grytten_2019}.
	
	The construction of VG models is more computationally resource-intensive
	than the construction of dBG. 
	Given 100 human genomes, the respective
	dBG can be built in a few hours \cite{minkin_2017,yu_2022}, 
	while the construction of VG requires several days
	\cite{garrison_2018,garrison_2023,hickey_2022,li_2020}. 
	Moreover, there are many possible variation graphs 
	that represent a particular collection of genomes. 
	Two completely uninformative extremes are: 
	an empty graph with each node labeled with one of the genome sequences, 
	and a graph with 4 nodes labeled with single letters A, C, G, T. 
	Existing tools tend to strike a balance between these extremes, 
	but to the best of our knowledge, neither minimum requirements nor global optimization criteria 
	have been proposed for variation graphs.
	
	The aim of the current paper is to fill this gap.
	We propose an axiomatization of the desired properties of variation graphs. 

	Moreover, we show the relationships between 
	VGs satisfying proposed criteria and related dBGs. 
	Finally, we show how these relationships may be used 
	to transform de Bruijn graphs into variation graphs.
	
	\section{Representing string collections with graphs}
	
	In this section we propose the concept of string graph 
	-- a common abstraction of de Bruijn and variation graphs.
	Then we formalize the notion of a graph representing given set of strings
	and define postulated properties of such representations.
	Finally, we show that these properties are always satisfied in de Bruijn graphs and impose strict constraints on the structure of variation graphs.

	\subsection{String graphs}
	
	A \emph{string graph} is a tuple $G=\langle V, E, l \rangle$, where:
	\begin{itemize}
		\item $V$ is a set of vertices,
		\item $E\subseteq V^2$ is a set of directed edges,
		\item $l: V \to \Sigma^+$ is a function labeling
		vertices with non-empty strings over alphabet $\Sigma$.
	\end{itemize}

	A \emph{path} in a string graph is a sequence of vertices $\langle v_1,\ldots,v_k\rangle$
	such that $\langle v_{i}, v_{i+1}\rangle \in E$ for every $i\in\{1,\ldots,k-1\}$.
	The set of all paths in $G$ will be denoted by $\mathcal{P}(G)$.
	Given path $p=\langle v_1,\ldots,v_k\rangle$, 
	the set of \emph{intervals} of $p$ is defined by formula
	$Int(p)=\{\langle i,j\rangle\;|\;1\le i\le j\le k\}$.
	Given two intervals 
	$\langle i,j\rangle,\langle i',j'\rangle\in Int(p)$,
	we say that $\langle i,j\rangle$ is 
	a \emph{subinterval} of $\langle i',j'\rangle$
	if $i\ge i'$ and $j\le j'$.
	A \emph{subpath} of $p$ defined on interval $\langle i,j\rangle$ 
	is a path $p[i..j]=\langle v_i,\ldots,v_j\rangle$.
	We use similar terminology and notation for strings:
	given string $S$, the set of its \emph{intervals} 
	is denoted by $Int(S)$,
	and $S[i..j]$ denotes the substring of $S$ indicated by interval $\langle i,j\rangle$.
	
	In order to represent strings longer than labels of single vertices,
	the labeling function $l$ must be extendable to function $\hat{l}$ defined on paths.
	The extension should be \emph{subpath-compatible} in the following sense:
	every path $p$ should induce an injective function $\Psi_p:Int(p)\to Int(\hat{l}(p))$
	satisfying the condidtion
	$$\Psi_p(i,j)=\langle i',j'\rangle \Rightarrow \hat{l}(p[i..j])=\hat{l}(p)[i'..j'].$$
	In the following subsection we use this concept to introduce 
	the definition of a graph representing a set of strings 
	and formulate appropriate properties of such representation.
	In subsequent subsections we describe two different its implementations
	that are realized in de Bruijn graphs and variation graphs, respectively.

	\subsection{Representations of collections of strings}\label{representation}
	
	Given a set of strings  $\mathcal{S}=\{S_1,\ldots,S_n\}$, 
	a string graph $G$ with subpath-compatible labeling extension $\hat{l}$ 
	and $\pi: \mathcal{S}\to\mathcal{P}(G)$,
	we say that $\langle G, \pi\rangle$ is a \emph{representation} of $\mathcal{S}$ 
	iff the following conditions are satisfied:
	\begin{itemize}
		\item $\hat{l}(\pi(S_i))=S_i$ for every $i\in\{1,\ldots,n\}$,
		\item every vertex in $G$ occurs in some path $\pi(S_i)$,
		\item every edge in $G$ joins two consecutive vertices in some path $\pi(S_i)$.
	\end{itemize}
	
	We define the set of \emph{positions} in $\pi$ as
	$Pos(\pi)=\{\langle i,j\rangle \;|\; 1\le i\le n \wedge 1\le j\le|\pi(S_i)|\}$.
	The set of \emph{$\pi$-occurrences} of a vertex $v$ 
	is defined as $Occ_{\pi}(v)=\{\langle i,j\rangle\in Pos(\pi) | \pi(S_i)[j]=v\}$. 
	The injections $\Psi_{\pi(S_i)}:Int(\pi(S_i))\to Int(S_i)$ 
	will be denoted by $\Psi_i$ for short.

	Let $S_i,S_{i'}$ be two (not necessarily different) strings from $\mathcal{S}$ and assume that $S_i[p..p+k-1]=S_{i'}[p'..p'+k-1]$ is a common $k$-mer of $S_i$ and $S_{i'}$. 
	We say that this common $k$-mer is \emph{reflected} by a common subpath 
	$\pi(S_i)[q..q+m]=\pi(S_{i'})[q'..q'+m]$ 
	of $\pi(S_i)$ and $\pi(S_{i'})$ iff
	$\langle p,p+k-1\rangle$ is a subinterval of $\Psi_i(q,q+m)$ and
	$\langle p',p'+k-1\rangle$ is a subinterval of $\Psi_{i'}(q',q'+m)$.

	We say that $\langle G, \pi\rangle$ represents $\mathcal{S}$ 
	\emph{$k$-completely} iff  
	all common $k$-mers in $\mathcal{S}$
	are reflected by respective subpaths.
	
	We say the pair of $\pi$-occurrences $\langle i,j\rangle, \langle i',j'\rangle$ of a vertex $v$ is:
	\begin{itemize}
		\item \emph{directly $k$-extendable} iff
		$\pi(S_i)[j-m..j+m']=\pi(S_{i'})[j'-m..j'+m']$
		for $m,m'\ge0$ satisfying $|\hat{l}(\pi(S_i)[j-m..j+m'])|\ge k$, 
		i.e. these occurrences extend to intervals of $\pi(S_i)$ and $\pi(S_{i'})$, respectively,
		indicating their common subpath
		labeled with a string of length $\ge k$, 
		\item \emph{$k$-extendable} if there is a sequence of occurrences of $v$
		that starts from $\langle i,j\rangle$, ends at $\langle i',j'\rangle$ 
		and each two consecutive occurrences in that sequence 
		are directly $k$-extendable.
	\end{itemize}
	We say that $\langle G, \pi\rangle$ represents $\mathcal{S}$ \emph{$k$-faithfully} 
	if every pair of occurrences of a vertex is $k$-extendable.
	
	Note that the $k$-completeness property specifies, 
	which fragments of $\mathcal{S}$-strings must be unified in the representation, 
	while $k$-faithfulness states 
	that we cannot unify anything that is not 
	a consequence of  $k$-completeness.
	Both properties 
	are illustrated in Fig. \ref{fig_k-prop}.
	
	\begin{figure}
		\centering
		\includegraphics [scale=0.5]{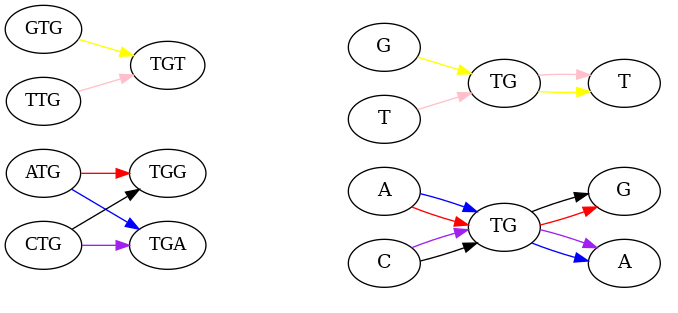}
		\caption{Two examples of a $3$-complete and $3$-faithful 
			representations of a set of strings
			$\{GTGT, TTGT, ATGG, ATGA, CTGG, CTGA \}$ - de Bruijn graph on the left and variation graph on the right.
			Different colors of edges correspond to different paths in $\pi$.  
			Note that in VG the common $2$-mer $TG$ from sequences  
			$ATGG$ (red), $ATGA$ (blue), $CTGG$ (black) and $CTGA$ (purple) 
			can be represented by a common vertex, 
			because its occurrences 
			on blue and red path can be extended to the common path labeled by  $3$-mer  $ATG$, 
			the occurrences on red and black path can be extended to a path labeled by $TGG$ 
			and the occurrences on black and purple can be extended to the path labeled by $CTG$. 
			On the other hand, 
			the occurrences of $2$-mer $TG$ in sequences $GTGT$ (yellow) and  $TTGT$ (pink) 
			cannot be represented by this vertex, 
			since its occurrences in yellow and pink paths 
			are not extendable to a path labeled by a common $3$-mer 
			with any other occurrence on the rest of the paths.}\label{fig_k-prop}
	\end{figure}

	\subsection{De Bruijn graphs}
	
	A \emph{de Bruijn graph} of length $k$
	is a string graph satisfying the following conditions:
	\begin{itemize}
		\item $|l(v)|=k$ for every $v\in V$,
		\item $l(v)=l(w)\Rightarrow v=w$ for all $v,w\in V$,
		\item $l(v)[2..k]=l(w)[1..k-1]$ for every $\langle v,w\rangle\in E$.
	\end{itemize} 
	In other words: vertices are labeled with unique $k$-mers 
	and edges connect vertices having labels overlapping 
	with $k-1$ characters.
	
	We define the extension $\hat{l}$ of the 
	de Bruijn graph labeling function in the following way:
	the labeling of a path $p=\langle v_1,\ldots,v_m\rangle$ is a concatenation 
	of the labels of its vertices with deduplicated overlaps, i.e. 
	$\hat{l}(p)=l(v_0)\cdot l(v_1)[k]\cdot\ldots\cdot l(v_m)[k]$.
	
	Consider $\langle i,j\rangle\in Int(p)$.
	The label of the subpath of $p$ indicated by $\langle i,j\rangle$
	satisfies the equation
	$\hat{l}(p[i..j])
	=\hat{l}(p)[i..j+k-1]$,
	so the function $\Psi_p(i,j)=\langle i,j+k-1\rangle$
	ensures subpath-compatibility.
	
	\begin{proposition}
		Given set of strings $\mathcal{S}=\{S_1,\ldots,S_n\}$ 
		such that $|S_i|\ge k$ for every $i\in\{1,\ldots,n\}$, 
		there is a unique up to isomorphism 
		representation of $\mathcal{S}$ by a de Bruijn graph of length $k$.
		Moreover this representation is $k$-complete and $k$-faithfull.
	\end{proposition}
	
	\begin{proof}
		The set of vertices of a dBG of length $k$ representing $\mathcal{S}$
		is actually determined by the set of different $k$-mers in $\mathcal{S}$.
		A mapping between the sets of vertices of two such graphs 
		that preserves the labels of those vertices, must be an isomorphism.
		Uniqueness of vertex labels implies also $k$-completeness.
		Finally, $k$-faithfulness is obvious from the fact 
		that the graph has no vertices with labels shorter than $k$.
	\end{proof}

	\subsection{Variation graphs}
		
	In \emph{variation graphs} 
	vertices may be labeled with strings of any length. 
	The extension of the labeling function to paths is defined as 
	the concatenation of the labels of consecutive vertices,
	i.e. $\hat{l}(p)=l(v_1)\cdot\ldots\cdot l(v_m)$
	for $p=\langle v_1,\ldots,v_m\rangle$.
	
	Consider a path $p=\langle v_1,\ldots,v_m\rangle$
	and let $s_i=\sum_{j=1}^{i}|l(v_i)|$ for $i\in\{0,\ldots,m\}$.
	The label of the subpath of $p$ indicated by $\langle i,j\rangle\in Int(p)$
	satisfies the equation
	$\hat{l}(p[i..j])
	=\hat{l}(p)[s_{i-1}+1..s_j]$,
	so the function $\Psi_p(i,j)=\langle s_{i-1}+1,s_j\rangle$
	ensures subpath-compatibility.
	When $|l(v)|=1$ for every vertex $v$, the graph is called \emph{singular}.
	In this case $Int(p)=Int(\hat{l}(p))$, $s_i=i$ for all $i\in\{0,\ldots,m\}$
	and the above formula simplifies to $\Psi_p(i,j)=\langle i,j\rangle$.
	
	Two representations of the set of strings $\mathcal{S}$ are \emph{equivalent} 
	iff they reflect exactly the same $m$-mers for all $m>0$.
	Below we show that $k$-completeness and $k$-faithfulness
	properties determine the structure 
	of the VG-representation up to equivalence.
	
	\begin{lemma}\label{lem:equiv}
		Assume that variation graph representations 
		$\langle G,\pi\rangle$ and $\langle G',\pi'\rangle$
		of a set of strings $\mathcal{S}=\{S_1,\ldots,S_n\}$
		are $k$-complete.
		Let $\langle i_1,j_1\rangle, \langle i_2,j_2\rangle$
		be a $k$-extendable pair of $\pi$-occurrences of a vertex $v$.
		Then every vertex $v'$ of $G'$ having 
		such $\pi'$-occurrence $\langle i_1,j_1'\rangle$
		that $S_{i_1}$-intervals 
		$\Psi_{\pi'(S_{i_1})}(j_1',j_1')$ and
		$\Psi_{\pi(S_{i_1})}(j_1,j_1)$ overlap,
		has also such $\pi'$-occurrence $\langle i_2,j_2'\rangle$
		that $S_{i_2}$-intervals 
		$\Psi_{\pi'(S_{i_2})}(j_2',j_2')$ and
		$\Psi_{\pi(S_{i_2})}(j_2,j_2)$ overlap.
	\end{lemma}
	
	\begin{proof}
		From the definition of $k$-extendability there exists
		a sequence of $\pi$-occurrences of $v$,
		in which each two consecutive occurrences 
		extend to intervals indicating common subpath
		labeled with a string of length $\ge k$.
		Every such extension reflects a common $k$-mer in $\mathcal{S}$. 
		In the case of the first pair of $\pi$-occurrences of $v$,
		the $S_{i_1}$-interval indicating the occurrence
		of the $k$-mer 
		overlaps $\Psi_{\pi'(S_{i_1})}(j_1',j_1')$.
		From $k$-completeness of $\langle G',\pi'\rangle$, 
		by induction, $v'$ has $\pi'$-occurrences 
		satisfying such condition for 
		every $\pi$-occurrence of $v$ in the sequence,
		in particular for $\langle i_2,j_2\rangle$.
	\end{proof}
	
	\begin{theorem}\label{thm:equiv}
		Assume that variation graph representations 
		$\langle G,\pi\rangle$ and $\langle G',\pi'\rangle$
		of a set of strings $\mathcal{S}=\{S_1,\ldots,S_n\}$
		are $k$-complete and $k$-faithfull.
		Then $\langle G,\pi\rangle$ and $\langle G',\pi'\rangle$
		are equivalent.
	\end{theorem}
	
	\begin{proof}
		Assume that the common $m$-mer $S_{i_1}[p..p+m-1]=S_{i_2}[p'..p'+m-1]$ 
		of $S_{i_1}$ and $S_{i_2}$ 
		is reflected in one of the representations, 
		say $\langle G,\pi\rangle$.
		Thus there exists a common subpath of
		$\pi(S_{i_1})$ and $\pi(S_{i_2})$ 
		covering the considered occurrences of the $m$-mer.
		For each vertex on this subpath the pair of its
		occurrences on $\pi(S_{i_1})$ and $\pi(S_{i_2})$ 
		is $k$-extendable due to $k$-faithfulness of
		$\langle G,\pi\rangle$.
		For every position $\langle i_1,j'\rangle$ 
		on the minimal subpath of $\pi'(S_{i_1})$
		covering $S_{i_1}[p..p+m-1]$,
		$S_{i_1}$-interval 
		$\Psi_{\pi'(S_{i_1})}(j',j')$ overlaps
		$\Psi_{\pi(S_{i_1})}(j,j)$,
		where $\langle i_1,j\rangle$ is one of the positions
		on the $\pi(S_{i_1})$-subpath covering the $m$-mer.
		Thus, due to Lemma \ref{lem:equiv}, all vertices 
		on the minimal subpath of $\pi'(S_{i_1})$
		covering $S_{i_1}[p..p+m-1]$ 
		have corresponding occurrences on $\pi'(S_{i_2})$,
		which, combined together,
		form a subpath covering $S_{i_2}[p'..p'+m-1]$.
		Therefore the considered common $m$-mer is reflected
		in representation $\langle G',\pi'\rangle$ too.
	\end{proof}

	The above theorem shows that 
	whenever a $k$-complete and $k$-faithful 
	variation graph representation of a given set of strings
	exists,
	it is determined up to equivalence. 
	In the next section we complement this result by
	showing how such variation graph can be built.
	
	\section{Graph transformation}
	
	In this section we show how to transform a de Bruijn graph
	representing a given set of strings 
	into a corresponding variation graph.
	The transformation algorithm consists of 3 steps
	(see Figure \ref{fig_steps}):
	\begin{enumerate}
		\item Split -- 
		conversion of vertices of the de Bruijn graph
		to unbranched paths with each vertex
		labeled with a single character.
		\item Merge -- a series of local modifications
		that merge incident edges inherited from the de Bruijn graph.
		\item Collapse -- a series of local modifications
		that removes isolated edges inherited from the de Bruijn graph.
	\end{enumerate}
	
	\begin{figure}[t]
		\centering
		\includegraphics[width=\textwidth]{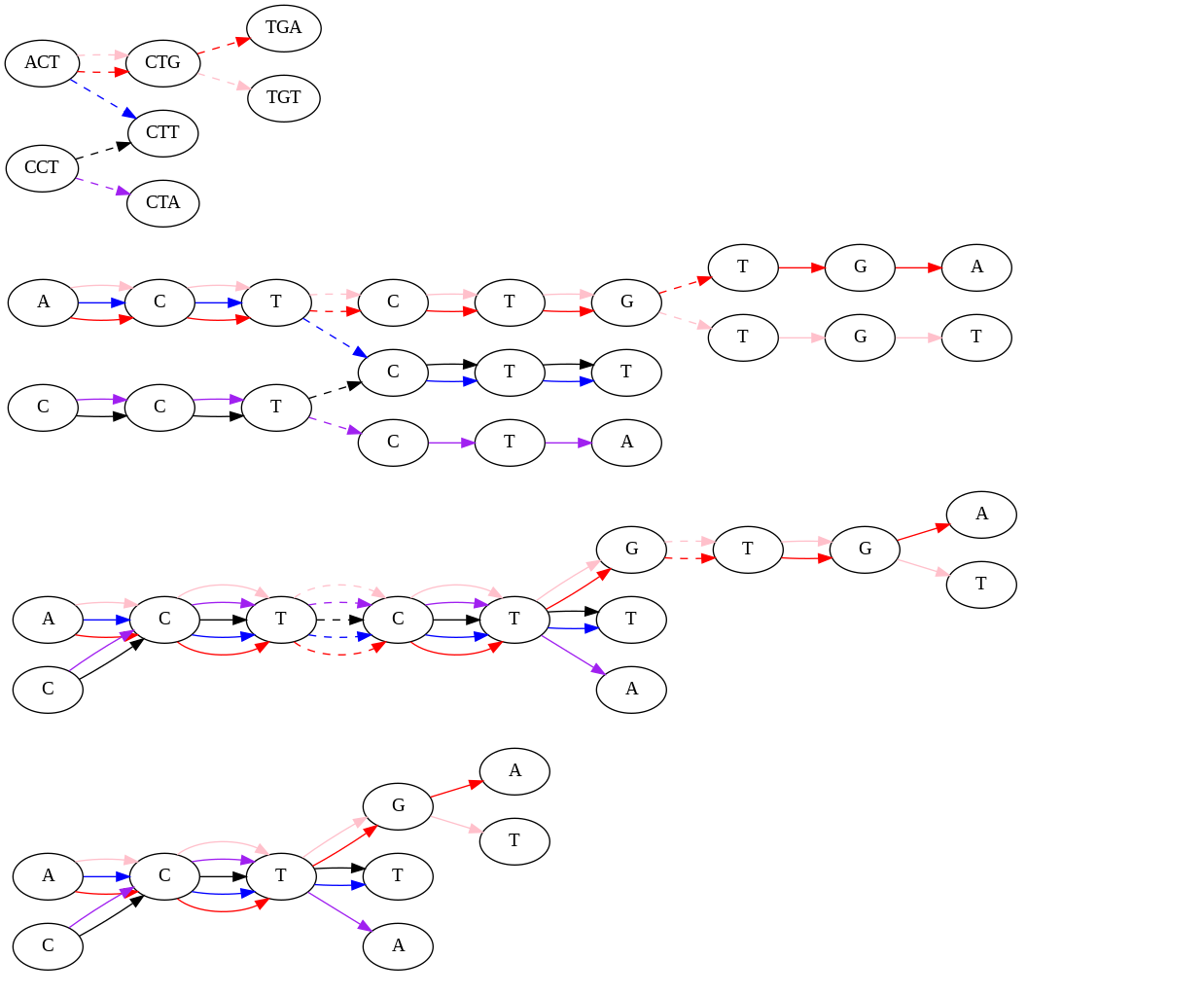}
		\caption{Steps of the graph transformation algorithm 
			for the graph representation of	the set of strings			 
			$S= \{ACTGA, ACTGT, ACTT, CCTT, CCTA\}$. 
			From top: input de Bruijn graph for $k=3$
			and transition graphs resulting from 
			Split, Merge and Collapse transformations,
			respectively. 
			Solid lines represent V-edges, dashed lines represent B-edges.}\label{fig_steps}
	\end{figure}
	
	We begin this section by introducing 
	\emph{transition graph}
	-- another type of string graph that will be used 
	in intermediate stages of the transformation.
	Then we describe consecutive transformation steps
	and show that they preserve desirable properties.
	Finally, we show that the whole transformation 
	yields a singular variation graph that is 
	both $k$-complete and $k$-faithful.
	
	\subsection{Transition graphs}
	
	A \emph{transition graph} of length $k$ over an alphabet $\Sigma$ 
	is a string graph $G=\langle V,E,l\rangle$, in which 
	$l:V\to\Sigma$ (i.e. every vertex is labeled with a single character) 
	and $E$ is a disjoint union of two subsets: 
	\begin{itemize}
		\item $E_V$ -- the set of \emph{variation edges} (or \emph{V-edges} for short),
		\item $E_B$ -- the set of \emph{de Bruijn edges} (or \emph{B-edges} for short),
	\end{itemize}
	A path is called \emph{V-path} if all its edges are V-edges.
	Labeling function $l$ is extended to V-paths as in variation graphs, 
	i.e. $\hat{l}(\langle v_1, \ldots, v_m\rangle)=  l(v_1)\cdot\ldots\cdot l(v_m)$
	for a V-paths $\langle v_1, \ldots, v_m\rangle$.
	
	Every path $p$ with $n$ $B$-edges can be split into a sequence 
	of $n+1$ \emph{maximal sub-V-paths} $\langle p_0,\ldots, p_n\rangle$
	such that for every $i\in\{1,\ldots,n\}$
	the last vertex of $p_{i-1}$ and the first vertex of $p_i$ 
	are connected by a $B$-edge.
    
	When this decomposition satisfies the following conditions:
	\begin{itemize}
		\item $|p_0|\ge k-1$, $|p_n|\ge k-1$ and $|p_i|\ge k$ for $i\in\{1,\ldots,n-1\}$,
		\item $\hat{l}(p_{i-1})[|p_{i-1}|-k+2..|p_{i-1}|]=\hat{l}(p_{i})[1..k-1])$
		for $i\in\{1,\ldots,n\}$,
	\end{itemize}
	the extension of labeling function is defined by formula 
	$\hat{l}(p)=\hat{l}(p_0)\cdot\hat{l}(p_1)[k..|\hat{l}(p_1)|]\cdot\ldots\cdot\hat{l}(p_n)[k..|\hat{l}(p_n)|]$.
	If moreover $|p_0|\ge k$ and  $|p_n|\ge k$,
	we call $p$ a \emph{valid} path.
	Interval $\langle i,j\rangle\in Int(p)$ is called \emph{valid}
	iff $p[i..j]$ is a valid path.
	The set of all valid intervals of a valid path $p$ will be denoted by $VInt(p)$.
	Since the labeling function is not defined on all possible paths,
	the notion of subpath-compatibility is adapted accordingly:
	the domain of $\Psi_p$ injection is restricted to $VInt(p)$.
	
	Let $p=\langle v_1, \ldots, v_m\rangle$ be a valid path 
	and let $S=\hat{l}(p)$.
	Let $b_i$ denote the number of $B$-edges in $p$ preceding its $i$-th vertex.
	The function $\psi_p:\{1,\ldots,m\}\to\{1,\ldots,|S|\}$
	defined by the formula $\psi_p(i)=i-(k-1)b_i$
	indicates residues in $S$ corresponding to particular vertices in $p$,
	i.e. the residues for which the condition $S[\psi_p(i)]=l(v_i)$ must be satisfied.
	Thus the function $\Psi_p:VInt(p)\to Int(S)$ defined by formula
	$\Psi_p(i,j)=\langle\psi_p(i),\psi_p(j)\rangle$ 
	ensures subpath-compatibility.

	All the definitions from Section \ref{representation} apply to transition graphs 
	with the only restriction that strings can be represented by valid paths only.
	
	Let $S_{i},S_{i'}$ be two (not necessarily different) strings from $\mathcal{S}$ and assume that 
	their representations
	have a common B-edge 
	$\pi(S_{i})[j..j+1]=\pi(S_{i'})[j'..j'+1]$.
	We say that these B-edge occurrences are \emph{consistent} 
	iff they are extandable to a common subpath
	$\pi(S_{i})[j-k+2..j+k-1]=\pi(S_{i'})[j'-k+2..j'+k-1]$.
	Representation $\langle G,\pi\rangle$ 
	is \emph{consistent} iff all common occurrences of $B$-edges are consistent.
	When this condition holds 
	for every pair of valid paths in $G$
	(i.e. not necessarily representing $\mathcal{S}$-sequences), 
	we say that $G$ is \emph{consistent}.

	\subsection{Transformation 1: Split}
	
	The Split operation transforms a dBG
	representation 
	$\langle G, \pi\rangle$ of a string collection
	$\mathcal{S}=\{S_1,\ldots,S_n\}$
	into a transition graph representation $\langle G', \pi'\rangle$
	by splitting each vertex $v$ of $G$ 
	into $k$ vertices labeled with single characters.
	More formally, $G'=\langle V', E_B, E_V, l'\rangle$, where:
	\begin{itemize}
		\item $V'= \underset{v\in V}{\bigcup}\{v_1,\ldots,v_k\}$,
		\item $E_V= \underset{v\in V}{\bigcup}
		\{\langle v_1,v_2\rangle,\ldots,\langle v_{k-1},v_k\rangle\}$,
		\item $E_B= \{\langle v_k,w_1\rangle: \langle v,w\rangle\in E\}$,
		\item $l'(v_i)=l(v)[i]$ for all $v\in V$ and $i\in\{1,\ldots,k\}$,
	\end{itemize}
	and paths $\pi'(S_i)$ are constructed from $\pi(S_i)$ by replacing 
	each vertex $v$ with a sequence $v_1,\ldots,v_k$.
	
	\begin{lemma}\label{lem_split}
		$G'$ is a consistent transition graph of length $k$ and  
		$\langle G', \pi'\rangle$ represents $\mathcal{S}$ 
		$k$-completely and $k$-faithfully.
	\end{lemma}
	
	\begin{proof}
		Consistency follows from the fact that every $B$-edge
		$\langle v_k,w_1\rangle$ is preceded and followed by
		unbranched $V$-paths $\langle v_1,\ldots,v_k\rangle$ and 
		$\langle w_1,\ldots,w_k\rangle$, respectively.
		Every common $k$-mer $W$ in $\mathcal{S}$ 
		was reflected in $\langle G, \pi\rangle$
		by the unique vertex $v\in V$ such that $l(v)=W$, 
		and thus is also reflected in $\langle G', \pi'\rangle$
		by the path $\langle v_1,\ldots,v_k\rangle$,
		which proves that the representation is	$k$-complete.
		Finally, $k$-faithfullness is due to the fact that 
		every occurrence of a vertex $v_i$ in any valid path 
		extends to the same subpath $\langle v_1,\ldots,v_k\rangle$.
	\end{proof}

	\subsection{Transformation 2: Merge}

	We define $F(G)$ as a set of vertices in $G$ with at least two outgoing $B$-edges, and $B(G)$ as a set of vertices in $G$ with at least two incoming $B$-edges.
	The second step of our algorithm consists of 
	a series of local graph modifications 
	affecting either $V$-paths following elements of $F(G)$
	or $V$-paths preceding elements of $B(G)$,
	called \emph{forward merge} and \emph{backward merge}, respectively. 
	These operations are symmetric, so 
	below we describe the forward merge operation only.
	
	Let $v$ be an element of $F(G)$. 
	As we know from the consistency condition, 
	each $B$-edge starting from $v$ is 
	preceded by the same $V$-path
	$\langle v_1, \ldots, v_{k-2}, v \rangle$
	and followed by another $V$-path of length $k-1$
	and the same sequence of labels.
	Let us denote all these paths as 
	$\langle w^1_1, \ldots, w^1_{k-1} \rangle,\ldots,\langle w^m_1, \ldots, w^m_{k-1} \rangle$.
	
	The forward merge operation for each $n \in \{1,\ldots, k-1\}$ 
	merges all vertices $w^1_n,\ldots,w^m_n$ into a single vertex $w_n$ 
	(note that all vertices $w^1_n,\ldots,w^m_n$
	have the same label, namely $l(v_n)$).
	Edges incident with merged vertices are replaced
	with edges joining respective new vertices.
	Consequently, all $B$-edges outgoing from $v$
	are merged into a single $B$-edge $\langle v,w_1\rangle$.
	
	The transformed paths $\pi'(S_i)$ are obtained from $\pi(S_i)$ 
	by replacing all occurrences 
	of vertices $w^1_n,\ldots,w^m_n$ with $w_n$. 
	
	\begin{lemma}\label{lem_merge}
		The Merge operation preserves the consistency of the graph,
		as well as the $k$-completeness and $k$-faithfulness of the representation.
	\end{lemma}
	
	\begin{proof}
		Consistency follows from the fact that merged
		$B$-edges are followed in a new graph by a unique path 
		$\langle w_1, \ldots, w_{k-1}\rangle$.
		
		If a common $k$-mer is reflected in $ \langle G, \pi \rangle$ by a $V$-path $p$ containing the vertex $w^i_n$, then after the merge, it is reflected in $ \langle G', \pi' \rangle$  by a path obtained from $p$ by replacing $w^i_n$ with $w_n$. Therefore, the representation is still $k$-complete.
		
		Let $\langle i,j \rangle , \langle i',j' \rangle$ be 
		two $\pi'$-occurrences of the same vertex 
		$u'$ in paths $\pi'(S)$, i.e $u'=\pi'(S_i)[j]=\pi'(S_{i'})[j']$.
		If corresponding paths in $\pi$ also have 
		the same vertex $u$ on respective positions 
		(i.e. $u=\pi(S_i)[j]=\pi(S_{i'})[j']$),
		the sequence supporting the $k$-extendability 
		of these $\pi$-occurrences of $u$
		supports the $k$-extendability 
		of the $\pi'$-occurrences of $u'$ too.
		The only case where the $\pi$-paths 
		have different vertices on these positions 
		is when $\pi(S_i)[j]=w^m_n$ and $\pi(S_{i'})[j']= w^{m'}_n$ for $m\neq m'$.
		
		If both $\pi'(S_i)[j]$ and $\pi'(S_{i'})[j']$ 
		can be extended to the common subpaths 
		of the form $\langle v, w_1, w_2, \ldots, w_{k-1} \rangle$, 
		we have $v=\pi'(S_i)[j-n]=\pi'(S_{i'})[j'-n]$
		and this pair of $\pi'$-occurrences of $v$
		is $k$-extendable, 
		because vertex $v$ was unaffected by the transformation.
		Thus there exists a sequence 
		$\langle i,j-n \rangle=\langle i_0,j_0 \rangle,
		\langle i_1,j_1 \rangle,\ldots,
		\langle i_l,j_l \rangle=\langle i',j'-n \rangle$
		of $\pi'$-occurrences of $v$
		supporting the $k$-extendability of this pair.
		for some $1<l'<l$,
		the occurrence $\langle i_{l'},j_{l'} \rangle$
		can be skipped in this sequence, 
		because in this case occurrences
		$\langle i_{l'-1},j_{l'-1} \rangle$
		and $\langle i_{l'+1},j_{l'+1} \rangle$
		must be directly $k$-extendable.
		Therefore we can assume without loss of generality
		that $j_{l'}<|\pi'(S_{i_{l'}})|$
		for all $1\le l'\le l$.
		Thus all these occurrences of $v$
		are followed in paths 
		$\pi(S_{i_1}),\ldots,\pi(S_{i_1})$ 
		by occurrences of vertices
		$w_1, w_2, \ldots, w_{k-1}$.
		Consequently, the sequence
		$\langle i_0,j_0+n \rangle,
		\ldots,
		\langle i_l,j_l+n \rangle$
		supports the $k$-extendability 
		of the pair of 
		$\pi'$-occurrences
		$\langle i,j \rangle,\langle i',j' \rangle$ 
		of vertex $w_n$.

		The proof is concluded with the observation that
		each occurrence of $w_n$ originating from $w^m_n$ (respectively  $w^{m'}_n$)
		is $k$-extendable with some occurrence of this vertex belonging to a subpaths of the form 
		$\langle v, w_1, w_2, \ldots, w_{k-1} \rangle$, and all occurrences belonging to such subpath are $k-$extendable one with each other.
		
	\end{proof}

	\subsection{Transformation 3: Collapse}

	The third step of our transformation consists
	of a series of local modifications, 
	each of which removes one $B$-edge.
	
	Consider a $B$-edge $\langle v_{k-1},w_1\rangle$. 
	Due to the consistency of the representation 
	all the occurrences of this edge in paths representing $\mathcal{S}$ 
	extend to the same subpath
	$\langle v_1, \ldots, v_{k-1},w_1, \ldots, w_{k-1} \rangle$
	satisfying $l(v_1)=l(w_1),\ldots,l(v_{k-1})=l(w_{k-1})$.
	Collapse operation on the edge  $\langle v_{k-1},w_1\rangle$ consists of:
	\begin{itemize}
		\item merging each pair of vertices $v_n,w_n$
		into a new vertex $u_n$,
		\item removing the $B$-edge  $\langle v_{k-1},w_1\rangle$,
		\item replacing other edges incident with merged vertices with new edges incident with respective new vertices,
		\item replacing $\langle v_1, \ldots, v_{k-1},w_1, \ldots, w_{k-1} \rangle$ 
		with $\langle u_1, \ldots, u_{k-1}\rangle$
		in paths representing $\mathcal{S}$, 
		\item replacing other occurrences of $v_n$ and $w_n$ with $u_n$ in paths representing $\mathcal{S}$.
	\end{itemize}
	It may happen that a $\pi$-path traverses
	$B$-edge $\langle v_{k-1},w_1\rangle$ 
	several times such that $k-1$ vertices preceding 
	and $k-1$ vertices following this edge overlap,
	i.e. $v_1=w_{t+1},v_2=w_{t+2},\ldots,v_{k-1-t}=w_{k-1}$ for some $1<t<k-1$
	(see Fig. \ref{fig:cycle}).
	In this case:
	\begin{itemize}
		\item in addition, 
		vertices $u_n,u_{t+n},u_{2t+n},\ldots$
		for $1\le n<t$ are merged
		(as a transitive consequence of merging
		$v_n=w_{t+n},v_{t+n}=w_{2t+n},\ldots$),
		\item a subpath
		$\langle v_1, \ldots, w_{k-1} \rangle$ in $\pi(S_i)$
		traversing $B$-edge $\langle v_{k-1},w_1\rangle$ 
		$m$ times in this way 
		is replaced in $\pi'(S_i)$
		with subpath
		$\langle u_1, u_2 \ldots u_{k-1+(m-1)t} \rangle$, 
		where $u_{n+t}=u_{n}$.
	\end{itemize}

	We will denote by $U_n$ the set of all vertices 
	in representation $\langle G, \pi \rangle$ 
	that will be merged into a single vertex $u_n$ 
	in representation $\langle G', \pi' \rangle$ 
	resulting from the Collapse operation 
	on a $B$-edge $\langle v_{k-1},w_1\rangle$.
	
	It must contain $v_n$ and $w_n$, 
	but there can be more elements in this set.
	If $v_n=v_{n+t}$ (and thus path $\langle v_1, \ldots v_{k-1} \rangle$ has a cycle), 
	then also $w_{n+t}$ belongs to $U_n$ 
	and $U_n=U_{n+t}$.
	Similarly, if $w_n=w_{n+t}$ then $v_{n+t}$ has to belong to $U_n$.
	Finally, if $\langle v1, \ldots , v_{k-1} \rangle$ 
	and $\langle w_1, \ldots, w_{k-1} \rangle$ 
	share common vertices 
	$v_{n}=w_{m_1}, v_{m_1}=w_{m_2},v_{m_2}=w_{m_3} \ldots, v_{m_{r-1}}=w_{m_r}$ 
	for some positive integer $r$, 
	then $v_{m_1}, v_{m_2}, \ldots v_{m_r}$ 
	also belongs to $U_n$ 
	and $U_n=U_{m_1}=U_{m_2}= \ldots =U_{m_r}$.

.

	\begin{figure}
		\centering
		\includegraphics[width=\textwidth]{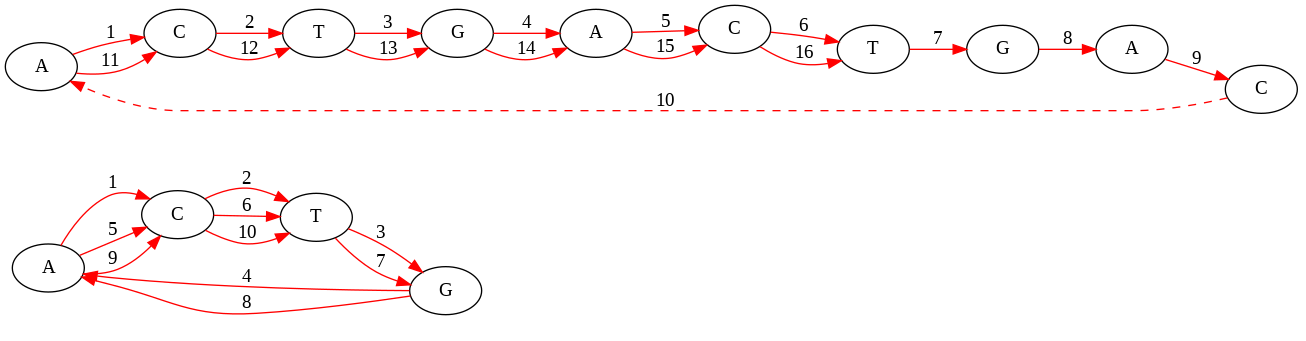}
		\caption{Example of Collapse transformation 
			applied to a $B$-edge,
			in which $k-1$ vertices following this edge 
			overlap with $k-1$ vertices preceding it. 
			$S=\{ACTGACTGACT\}$, $k=7$, 
			labels above edges indicate their order in the path.
			Top: $B$-edge 10 is 
			followed by $k-1$ vertices connected by edges 11-15
			and preceded by $k-1$ vertices connected by edges 5-9,
			the overlap forms a subpath with vertices $A$ and $C$
			connected by edge labeled with 5 before the $B$-edge and 15 after the $B$-edge.
			Bottom: after the Collapse operation 
			the overlapping subpath is merged with
			two other subpaths with vertices $A$ and $C$:
			preceding the $B$-edge and following it.
		}\label{fig:cycle}
	\end{figure}

	\begin{lemma}\label{lem_collapse}
		The Collapse operation preserves consistency, $k$-faithfulness and $k$-completeness of the representation.
	\end{lemma}
	
	\begin{proof}
		Consistency follows from the fact that Collapse
		does not modify the labels of the vertices.

		Similarly, $k$-completeness follows from the fact that 
		if a common $k$-mer was reflected in $\langle G, \pi\rangle$ by a valid $V$-path $p$, 
		then it is also reflected in $ \langle G', \pi' \rangle $ 
		by a path obtained from $p$ by replacing  occurrences of $v_n$ (or $w_n$, respectively) with  $u_n$.

		In order to prove $k$-faithfulness,	
		consider two occurrences $\langle i, j \rangle$ and $\langle i', j' \rangle$ of the same vertex 
		in representation $\langle G', \pi' \rangle$ originating from transformation of $\langle i, q \rangle$ and $\langle  i', q'\rangle$ in representation $\langle G, \pi \rangle$.
		If before the modification of the graph these paths had on respective positions 
		the same vertex $\pi(S_i)[q]= \pi(S_{i'})[q']$, 
		then their extensions are preserved and support the $k$-extendability 
		of  $\langle i, j \rangle$ and $\langle i', j' \rangle$.
		
		If paths in $\langle G, \pi \rangle$ had different vertices on respective positions, then both these vertices belong to the set $U_n$.
		
		First, consider situation when one of them was $v_n$ and the other was $w_n$. 
		We know that there is at least one path $\pi(S_i)$ 
		that traverses the whole subpath $\langle v_1, \ldots, v_{k-2} v_{k-1},w_1, w_2, \ldots, w_{k-1} \rangle$.  
		Thus in $\pi'(S_i)$ we have an occurrence of $u_n$
		that corresponds to occurrences of both $v_n$ and $w_n$ in $\pi(S_i)$ 
		on positions separated by $k-2$ vertices. 
		Therefore this occurrence of $u_n$ is $k$-extendable 
		with all other occurrences originating from $v_n$, 
		and also with all  other occurrences originating from $w_n$.

		For similar reasons all occurrences originating from $v_{n+t}$ 
		have to be $k$-extendable with the ones originating from $w_{n+t}$. 
		Thus if $v_n=v_{n+t}$, 
		all occurrences originating 
		from different vertices $w_n$ and $w_{n+t}$ 
		are also $k$-extendable 
		(and symmetrically: 
		all occurrences originating from different vertices $v_n$ and $v_{n+t}$
		are $k$-extendable  if $w_n=w_{n+t}$).
		
		The same argument applies to the case when 
		$v_{n}=w_{m_1}, v_{m_1}=w_{m_2}, v_{m_2}=w_{m_3}, \ldots, v_{m_{r-1}}=w_{m_r}$: 
		all occurrences of $u_n$ originating from $v_n$ 
		are $k$-extendable with those originating from $v_{m_1}$, 
		since $w_{m_1}=v_n$. 
		Because $w_{m_2}=v_{m_1}$,
		these occurrences are $k$-extendable with those originating from $v_{m_2}$. 
		Applying this reasoning repeatedly, 
		we conclude that
		all occurrences of $u_n$ originating from the vertex $v_n$ 
		are $k$-extendable with those originating 
		from any vertex $v_{m_s}$ for $s \in \{1, \ldots r \}$.
		
		Thus we showed that all occurrences of vertex $u_n$ in $\pi'(S_i)$ 
		originating from different vertices in $\pi(S_i)$ 
		are always $k$-extendable one with each other.

	\end{proof}

\subsection{Correctness of the algorithm}

We can now formulate the main result of this section.

\begin{theorem}\label{thm:build}
	Given a de Bruijn graph of length $k$
	representing a collection of strings $\mathcal{S}$,
	the transformation algorithm always terminates
	resulting in a $k$-complete and $k$-faithful
	variation graph representing $\mathcal{S}$.
\end{theorem}

\begin{proof}
	Algorithm terminates, because
	Split operation is applied only once for each vertex and every execution of Merge or Collapse
	reduce the number of $B$-edges in the graph. 

	The transformations ensure that the final
	transition graph has no $B$-edges, 
	so it is in fact a singular variation graph.
	Lemmas \ref{lem_split}-\ref{lem_collapse} guarantee $k$-completeness and $k$-faithfulness.
\end{proof}

Theorems \ref{thm:equiv} and \ref{thm:build} can be summarized by the following statement.

\begin{corollary}
	Let $\mathcal{S}=\{S_1,\ldots,S_n\}$ be a set of strings
	such that $|S_i|\ge k$ for every $i\in\{1,\ldots,n\}$, 
	Then the $k$-complete and $k$-faithful variation graph representation 
	of $\mathcal{S}$ exists and is unique up to equivalence.
\end{corollary}

\section{Conclusion}
In the current paper we proposed the concepts of 
$k$-completeness and $k$-faithfulness,
which may be used to express the desirable properties
of a graph representation of a collection of strings.
We showed that these properties
are always satisfied in de Bruijn graphs 
and determine up to equivalence the structure of variation graphs. 
Furthermore, we showed the relationship between 
variation graphs satisfying the above conditions 
and de Bruijn graphs.
This relationship can be used not only to 
efficiently build a variation graph 
representing a given set of sequences,
but also to provide a direct method 
of transferring annotations between both pangenome models.

The proposed axiomatization 
may be further developed.
For example, one can formulate properties 
of a variation graph
that express desirable differences 
from the structure associated with the
corresponding de Bruijn graph.

\subsubsection{Acknowledgements} 
	This research was funded in whole by National Science Centre, Poland,
	grant no. 2022/47/B/ST6/03154.

\bibliographystyle{splncs04}
\bibliography{main}

\end{document}